\documentclass[11pt]{amsart}

\usepackage{amsmath,amssymb,amsthm}
\usepackage{graphicx}
\usepackage{graphics}
\usepackage[bookmarks=false]{hyperref} 
\usepackage{multirow}
\usepackage[algoruled,vlined,linesnumberedhidden,norelsize]{algorithm2e2008} 

\usepackage[OT1]{fontenc}
\usepackage{ae,aecompl}

\newtheorem{mylemma}{Lemma}[section]
\newtheorem{mytheorem}[mylemma]{Theorem}

\theoremstyle{definition}
\newtheorem{myexample}[mylemma]{Example}
\newtheorem{mydefinition}[mylemma]{Definition}

\newcommand{\txt}[1]{\texttt{#1}}
\newcommand{\comb}[2]{{{#1} \choose {#2}}}
\newcommand{\scr}[1]{\mathcal{#1}}

\begin{document}

\title[Counting inequivalent monotone Boolean functions]{Counting inequivalent monotone Boolean functions}
\author{Tamon Stephen}
\author{Timothy Yusun}
\address{Department of Mathematics \\
Simon Fraser University \\
8888 University Drive \\
Burnaby, B.C. V5A 1S6\\
Canada
}
\email{tyusun@sfu.ca}
\email{tamon@sfu.ca}

\begin{abstract}
Monotone Boolean functions (MBFs) are Boolean functions $f: \{0,1\}^n \rightarrow \{0,1\}$ satisfying the monotonicity condition $x \leq y \Rightarrow f(x) \leq f(y)$ for any $x,y \in \{0,1\}^n$. The number of MBFs in $n$ variables is known as the $n$th Dedekind number. It is a longstanding computational challenge to determine these numbers exactly -- these values are only known for $n$ at most 8. Two monotone Boolean functions are inequivalent if one can be obtained from the other by renaming the variables. The number of inequivalent MBFs in $n$ variables was known only for up to $n = 6$. In this paper we propose a strategy to count inequivalent MBF's by breaking the calculation into parts based on the profiles of these functions. As a result we are able to compute the number of inequivalent MBFs in 7 variables. The number obtained is 490013148.
\end{abstract}

\maketitle

\begin{section}{Introduction}

A \emph{Boolean function on $n$ variables} (BF) is a function $f: \{0,1\}^n \rightarrow \{ 0,1 \}$. A \emph{monotone Boolean function} (MBF) additionally satisfies the condition $x \leq y \Rightarrow f(x) \leq f(y)$, for any $ x, y \in \{0,1\}^n$. We write $x \leq y$ if $x_i \leq y_i$ for all $i = 1, 2, \ldots, n$, and $x < y$ if $x \leq y$ and $x_i < y_i$ for some $i$. A BF is monotone if and only if it can be written as a combination of conjunctions and disjunctions only.

Since each input state in $\{0,1\}^n$ has two possible output states, there are a total of $2^{2^n}$ Boolean functions on $n$ variables. On the other hand, no exact closed form is known for the number of \emph{monotone} Boolean functions on $n$ variables. This number is usually denoted by $D(n)$, which is also called the $n$th \emph{Dedekind number}. These numbers are named after Richard Dedekind who defined them in \cite{gessamelte}. The first few values are given in Table~\ref{tableDN}, taken from \cite{oeis}. Currently, only values of $D(n)$ up to $n = 8$ are known. 

\begin{table}[!htbp]
\begin{center} 
	\begin{tabular}{| c | r | c |}
	\hline
	$n$ & $D(n)$ & Source\\
	\hline
	0 & 2 & \multirow{5}{*}{Dedekind, 1897} \\
	1 & 3 & \\
	2 & 6 & \\
	3 & 20 & \\
	4 & 168 & \\
	\hline
	5 & 7 581 & Church, 1940 \cite{churchD5} \\
	\hline
	6 & 7 828 354 & Ward, 1946 \cite{wardD6} \\
	\hline
	\multirow{2}{*}{7} & \multirow{2}{*}{2 414 682 040 998} & Church, 1965 \cite{churchD7} \\
	& & {\small (see also \cite{bermanD7})}  \\
	\hline
	\multirow{2}{*}{8} & \multirow{2}{*}{56 130 437 228 687 557 907 788} & Wiedemann, 1991 \cite{wiedemann} \\
	& &  {\small (see also \cite{fidytek})} \\
	\hline
\end{tabular}
\end{center} 
\caption{Known Values of $D(n)$, \txt{A000372}}
\label{tableDN}
\end{table}

Kisielewicz gives in \cite{kisielewicz} a logical summation formula for $D(n)$, however performing the computation using his summation has the same complexity as brute force enumeration of $D(n)$. (See \cite{korshunovsurvey}, e.g.) There are some asymptotic results concerning the behavior of $D(n)$, one of the earliest of which was a result of Kleitman in 1969, that $\log_2D(n) \sim {n \choose {\lfloor n/2 \rfloor}}$ \cite{kleitmanasymp}. So far, the most accurate one is given by Korshunov in \cite{korshunovsurvey}, given in Table~\ref{asymkorshunov}.

\begin{table}[!htbp]
\begin{center}
\begin{tabular}{| c |}
\hline
\\
$D(n) \sim 2^{n \choose {n/2}} \cdot \exp\left[ {n\choose{\frac{n}{2} - 1} }\left( 2^{-n/2} + n^2 2^{-n-5} - n 2^{-n-4}\right)  \right], \text{ for even $n$} $ \\
\\
\hline
\\
$D(n) \sim 2^{{n \choose {(n-1)/2}} + 1} \cdot \exp\left[ {n\choose{\frac{n-3}{2}} }\left( 2^{(-n-3)/2} - n^2 2^{-n-5} - n 2^{-n-3}\right)\right.$ \\
$\left.+ {n\choose{\frac{n-1}{2}} }\left( 2^{(-n-1)/2} - n^2 2^{-n-4}\right)  \right], \text{ for odd $n$}$ \\
\\
\hline
\end{tabular}
\end{center} 
\caption{Korshunov's Asymptotics} \label{asymkorshunov}
\end{table}

\begin{subsection}{Inequivalent MBFs}

We define an MBF $f$ to be \emph{equivalent} to another MBF $g$ if $g$ can be obtained from $f$ by a renaming of the variables. For example, the function $f(x_1,x_2,x_3) = (x_1 \wedge x_2) \vee (x_2 \wedge x_3)$ is equivalent to $g(x_1,x_2,x_3) = (x_1 \wedge x_2) \vee (x_1 \wedge x_3)$, since interchanging $x_1$ and $x_2$ sends $f$ to $g$. We write this as $f \sim g$. For brevity, from here on we write $x_i \wedge x_j$ as $x_ix_j$.

Let $R(n)$ be the number of equivalence classes defined by ``$\sim$'' among monotone Boolean functions on $n$ variables. As with $D(n)$, no closed form is known for $R(n)$, and in fact only the values up to $n=6$ have been computed. These appear to have been obtained by the straightforward method of listing all monotone Boolean functions on $n$ variables and then sorting them into equivalence classes. The known values from~\cite{oeis} are shown in Table~\ref{table:Rn}. 

\begin{myexample}
The five functions in $R(2)$ are: $f = 0$, $f=1$, $f = x_1$, $f = x_1 \vee x_2$, and $f = x_1x_2$. The functions in $D(2)$ are exactly these functions plus $f = x_2$, which is equivalent to $f = x_1$.
\end{myexample}

\begin{table}[!htbp]
\begin{center}
	\begin{tabular}{| c | r |}
	\hline
	$n$ & $R(n)$\\
	\hline
	0 & 2 \\
	1 & 3 \\
	2 & 5 \\
	3 & 10 \\
	4 & 30 \\
	5 & 210 \\
	6 & 16 353\\
	\hline
\end{tabular}
\end{center}
\caption{Known Values of $R(n)$, \txt{A003182}}\label{table:Rn}
\end{table}

\end{subsection}

\begin{subsection}{Terminology and Elementary Facts}\label{sec:term}

Define a \emph{minimal term} of an MBF $f$ to be an input $x \in \{0,1\}^n$ such that $f(x) = 1$ and $f(y) = 0$ if $y < x$. The minimal terms of a monotone Boolean function represent the ``smallest'' sets where the function equals one -- any input below $x$ evaluates to zero, and everything above evaluates to one by virtue of monotonicity. For example, the function $f(x_1,x_2,x_3) = x_1 \vee x_2x_3$ evaluates to one at $(1,0,0)$ and $(0,1,1)$, as well as at all vectors above $(1,0,0)$ and $(0,1,1)$, and evaluates to $0$ at all other vectors. Indeed, each MBF can be written as a disjunction of clauses, each representing one of its minimal terms.

MBFs can be classified according to the number of minimal terms they have. Call $D_k(n)$ the number of monotone Boolean functions on $n$ variables with $k$ minimal terms. Kilibarda and Jovovic~\cite{kj} derive closed form expressions for $D_k(n)$ for fixed $k = 4,5,\ldots,10$, these sequences are sequences \txt{A051112} to \txt{A051118} of \cite{oeis}.

A \emph{truth table} for a Boolean function is a row of zeros and ones which encodes the outputs of the function corresponding to every possible input state. To illustrate, the function on three variables $f(x_1,x_2,x_3) = x_1\vee x_2x_3$ has minimal terms $\left\{ \{1\},\{2,3\} \right\}$, and it has the following truth table:

\begin{table}[!htbp]
\begin{center}
	\begin{tabular}{| l | c | c | c | c | c | c | c | c |}
	\hline
	variables set to $1$ & $x_1,x_2,x_3$ & $x_2,x_3$ & $x_1,x_3$ & $x_3$ & $x_1,x_2$ & $x_2$ & $x_1$ & none \\
	\hline
	output states & $1$ & $1$ & $1$ & $0$ & $1$ & $0$ & $1$ & $0$ \\
	\hline
\end{tabular}
\end{center}
\caption{Truth table for the function $f(x_1,x_2,x_3) = x_1 \vee x_2x_3$}
\end{table}

Note that the input states on the top row are arranged in a reverse colexicographic (or \emph{colex}) order on $\{0,1\}^3$, defined as $x < y$ if $x \ne y$ and $x_k < y_k$ where $k = \max\{i : x_i \ne y_i\}$. Fixing this order, we write $f$ as \txt{11101010}. In general, any Boolean function on $n$ variables can be written as a $0$-$1$ string of length $2^n$ where each entry corresponds to an input state; we use this convention throughout this paper. The choice for the ordering has the nice property that the first $2^{n-1}$ of its entries all involve setting the variable $x_n$ to $1$, and the second half has inputs with $x_n = 0$.

The truth table form is the most compact way to represent general Boolean functions. For monotone Boolean functions, both the truth table and the minimal terms representation are useful for our purposes.

\begin{myexample}
The colexicographic order for two variables is $\{1,2\} > \{2\} > \{1\} > \{ \}$. The functions in $D(2)$, written in truth table form are $\{\txt{1111}$, $\txt{1110}$, $\txt{1100}$, $\txt{1010}$, $\txt{1000}$, $\txt{0000}\}$. 
\end{myexample}

\end{subsection}

\begin{subsection}{Motivation}
Monotone Boolean functions are interesting because many mathematical objects can be represented by MBFs. For instance, there is a one-to-one correspondence between MBFs in $D(n)$ and antichains in the set $2^{[n]}$, that is, pairwise incomparable subsets of the power set of $\{1,2,\ldots,n\}$. Specifically, the set of minimal terms of an MBF in $D(n)$ is an antichain in $2^{[n]}$. Since by Sperner's Theorem, any antichain on the $n$-set can have at most $n \choose \lfloor n/2 \rfloor$ elements, we have that any $n$-variable MBF can have at most $n \choose \lfloor n/2 \rfloor$ minimal terms.

The one-to-one correspondence between $n$-variable MBFs and Sperner hypergraphs is also well-known. In particular, each minimal term of an MBF maps to a hyperedge in the corresponding hypergraph, and because of pairwise incomparability, the hypergraph thus exhibits the Sperner property, that is, no hyperedge contains another.

Some other fields in which monotone Boolean functions appear include lattice theory \cite{ilya}, nonlinear signal processing \cite{ilya}, coding theory \cite{ito}, computational learning theory \cite{ilyawww}, game theory \cite{riquelme}, and computational biology (\cite{klamt},\cite{haus}).

For a comprehensive discussion of Boolean functions, see the recently-published book by Crama and Hammer \cite{cramahammer}.

\end{subsection}

\end{section}

\begin{section}{Computational Strategies}\label{sec:compstrat}

\begin{subsection}{Profiles of MBFs}

It is natural to refine the classification of monotone Boolean functions by number of minimal terms, and consider how many elements are contained in each of these terms. We define the notion of a profile formally as given in \cite{sperner}, and introduce some notation:

\begin{mydefinition}[Profile of an MBF]\label{def:prof}
Given an $n$-variable MBF $f$ where $f \not\equiv 1$, the \bf profile \rm of $f$ is a vector of length $n$ $(a_1,a_2,\ldots,a_n)$, where the $i$th entry is equal to the number of minimal terms of $f$ which are $i$-sets.
\end{mydefinition}

\begin{myexample}
The MBF \texttt{11111100} has minimal terms $\{2\},\{3\}$, and profile $(2,0,0)$, while the MBF \texttt{11111000} has minimal terms $\{1,2\}$, $\{3\}$, and profile $(1,1,0)$.
\end{myexample}

\begin{mydefinition}
Given profile vector $(a_1, a_2, \ldots, a_n)$, define $(a_1, a_2, \ldots, a_n)_D$ to be the number of monotone Boolean functions on $n$ variables with profile vector $(a_1, a_2, \ldots, a_n)$. Similarly define $(a_1, a_2, \ldots, a_n)_R$ for inequivalent monotone Boolean functions on $n$ variables.
\end{mydefinition}

Note that the number of variables $n$ is implicit in the profile vector -- it is just the length of the vector. Some relations between the profiles are described and proven in Lemma~\ref{proflemma}.

\begin{mylemma}
Assume that all profile vectors pertain to MBFs on $n$ variables, unless otherwise stated.
\begin{description}\label{proflemma}
\item[(A)] $(0,0,\ldots,a_i,\dots,0)_D = (0,0,\ldots,{n \choose i} - a_i,\ldots,0)_D$.
\item[(B)] If $a_1 > 0$, then $a_{n} = 0$ and \\ $(a_1, a_2, \ldots, a_{n-1},a_n)_D = (a_1 - 1, a_2, \ldots, a_{n-1})_{D_{n-1}}$.
\item[(C)] $(a_1, a_2, \ldots,a_{n-2}, a_{n-1},a_n)_D = (a_{n-1},a_{n-2},\ldots,a_2,a_1,a_n)_D$.
\end{description}
All these statements hold true when $D$ is replaced by $R$.
\end{mylemma}
\begin{proof}
The proof of each claim rests on the fact that there is a one-to-one correspondence between functions of the first type and functions of the second type, for the purposes of counting both $D(n)$ and $R(n)$.

\textbf{(A)}: Given an MBF with exactly $a_i$ $i$-sets as minimal terms, we can derive another MBF with minimal terms exactly the $\comb{n}{i} - a_i$ $i$-sets which were not taken in the first MBF. Furthermore, the images of any two equivalent functions under this correspondence will also be equivalent, under the same permutation.

\textbf{(B)}: If an $n$-variable MBF has a singleton set, say $\{n\}$, as a minimal term, then we know that the rest of its minimal terms cannot contain the element $n$. Hence $a_n = 0$. In addition, removing the term $\{n\}$, we are left with an $(n-1)$-variable MBF, with the profile $(a_1-1,a_2,\ldots,a_{n-1})$.

\textbf{(C)}: If $a_n = 1$, then $\{1,2,\ldots,n\}$ is the only minimal term. This implies that all the other $a_i$'s are zero, and hence the claim follows trivially.

If $a_n = 0$, assume that the minimal terms of an MBF $f$ with the given profile are $A_1,A_2,\ldots,A_k$. We know that none of the $A_i$'s are comparable, so it should follow that none of the sets $[n] - A_1, [n] - A_2, \ldots, [n] - A_k$ must be comparable as well. Hence the collection $\{[n] - A_j\}_{1\leq j \leq k}$ is the set of minimal terms of an MBF $g$ where the number of $i$-sets is equal to the number of $(n-i)$-sets in $f$. This proves that the profile of $g$ is $(a_{n-1},a_{n-2},\ldots,a_2,a_1,a_n)$.
\end{proof}

Lemma \ref{proflemma} is very useful in reducing the amount of computation that needs to be done to compute $D(n)$ or $R(n)$. For instance, when counting $R(7)$, instead of counting all $(0,0,k,0,0,0,0)$ for $k=1$ to $35$, we count the profiles up to $k = 17$. Part \textbf{(B)} enables us to refer back to $R(6)$ when considering profiles with a nonzero entry in the first position. The most useful is \textbf{(C)}, which effectively cuts all computation time in half.

\begin{subsubsection}{Generating Profiles}

Sequence \txt{A007695} on the OEIS gives the number of profile vectors for any $n$, and also outlines an algorithm that can be used to compute this number \cite{oeis}. We modify this algorithm to actually output the profiles that are being counted. We present this algorithm as Algorithm \ref{alg:prof}.

\SetAlFnt{\small}
\begin{algorithm} \label{alg:prof}
	\SetAlgoLined
	\KwIn{$n$}
	\KwOut{$P(n)$, the list of profiles of MBFs on $n$ variables}
	initialize $C := \txt{zeros}\left(n+1,{n \choose {\lfloor n/2 \rfloor}} + 1\right)$ \;
	$K := C$, $C(0,0) = 1$, $C(0,1) = 1$, $s = 2$ \;
	initialize $P(n) := \txt{zeros}\left(n+1,n\right)$ \;
	set the first column of $P(n)$ to the vector $[0 \ 1 \ 2 \ \cdots \ n]^T$ \;
	\txt{total}$\ := n + 1$ \;
	\For{$r = 1$ \KwTo $n$}{
		$d \leftarrow s$, $k \leftarrow r$, $j \leftarrow 0$, $s \leftarrow 0$ \;
		$x_{\max} = {n \choose r}$ \;
		\For{$x = 0$ \KwTo $x_{\max}$}{
			\ShowLn\If{$x \geq {k \choose r}$}{
				$k \leftarrow k + 1$ \;
			}
			\eIf{$x = 0$}{
				$K(r,x)=0$\;
				}{
				\ShowLn $K(r,x) = K(r-1,x-{{k-1} \choose {r}}) + {{k-1} \choose {r-1}}$ \;
			}
			\ShowLn\While{$j < K(r,x)$}{
				$d \leftarrow d - C(r-1,j)$ \;
				$j \leftarrow j + 1$ \;
			}
			$C(r,x)=d$\;
			$s \leftarrow s + d$\;
		}
		\If{$r \ne 1$}{
			\ShowLn \txt{recent} $=$ last $C(r,0) - C(r-1,0)$ rows of $P(n)$ \;
			\For{$x=1$ \KwTo $x_{\max}$}{
				\ShowLn\txt{candidates} $=$ rows of $\txt{recent}$ with $(r-1)$-st entry at least $K(r,x)$. \;
				Subtract $K(r,x)$ from the $(r-1)$-st column of \txt{candidates} \;
				Add $x$ to the $r$-th column of \txt{candidates} \;
				Append \txt{candidates} to $P(n)$ \;
				Update \txt{total} $\leftarrow$ \txt{total} $+$ size$(\txt{candidates})$. \;
			}
		}
	}
	Output $P(n)$. \;
\caption{Generating all profiles of MBFs on $n$ variables.}
\end{algorithm}

Algorithm \ref{alg:prof} uses a dynamic programming strategy, where the matrix $K(r,x)$ is built up from the previous values $K(r-1,x)$. In fact, the $(r,x)$-th entry of the matrix $K$ is a strict lower bound on the number of $(r-1)$-sets that any family of $x$ $r$-sets can contain. We use this information to generate the list of profiles $P(n)$. We prove these facts in Section~\ref{se:proof}.

Using the algorithm, we can compute the number of profiles of MBFs on $n$ variables, which is one less than Sequence A007695 on the OEIS, to account for the all-ones function. We show this in Table~\ref{table:numprof}.

\begin{table}[!htbp]
\begin{center}
	\begin{tabular}{| c | c | c | c |}
	\hline
	$n$ & Number of profiles & $n$ & Number of profiles\\
	\hline
	0 & 1 & 5 & 95 \\
	1 & 2 & 6 & 552\\
	2 & 4 & 7 & 5460 \\
	3 & 9 & 8 & 100708 \\
	4 & 25 & 9 & 3718353 \\
	\hline
\end{tabular}
\end{center}
\caption{Number of profiles for each $n$, from $n=0$ to $n=9$.}
\label{table:numprof}
\end{table}

\end{subsubsection}

\begin{subsubsection}{Using Profiles to Generate Functions}

A monotone Boolean function can be written uniquely as the disjunction of its minimal terms.  Thus we can generate all MBFs inductively beginning with profiles that have a single non-zero entry. These have a simple structure - each consists of $k$ subsets of size $i$, where $i$ is the index of the non-zero entry and $k$ is the value of that entry.  Suppose now there is a second non-zero entry with index $j$.  If that entry is 1, we can generate all functions for this profile by taking disjunctions between the functions in the first profile and all $j$-sets which are incomparable to each of them. If the $j$th entry is larger than 1, we just repeat the same steps until we generate the desired list of functions. Hence a typical computation might start with the profile $(0,7,0,0,0,0,0)$ and continue on to $(0,7,1,0,0,0,0)$, $(0,7,2,0,0,0,0)$, etc.
These lists are then in turn used as starting points for profiles with three non-zero terms, and so on.

Since we are computing inequivalent MBFs, we eliminate equivalent functions after each profile is generated. To do this we generate for each function obtained all 5040 equivalent functions, and take the ``least representative,'' by which we mean the lexicographically smallest function. This can be implemented quickly by keeping only the least representative as each permutation of the original MBF is generated.

With the goal of enumerating $R(7)$, we use Lemma \ref{proflemma} to eliminate profiles we do not have to compute. However, we have to compute some profiles more than once, as some profiles whose $R$-value can be obtained by symmetry are used as intermediate steps to compute other profiles. More details of the computation can be found in Section~\ref{sec:compute}.

The largest profile for $R(7)$ computed is that of $(0,0,7,7,0,0,0)$, having 5443511 functions. For further illustration, we give in Table~\ref{R5} the list of profiles in $R(5)$ and the number of functions in each one. Note the list does not include the all-ones function, which does not have a corresponding profile (by Definition~\ref{def:prof}).

\begin{table}[!htbp]
\begin{center}
	\begin{tabular}{| c | c | c | c | c | c | c | c |} 
	\hline
	Profile & \# & Profile & \# & Profile & \# & Profile & \#\\
	\hline
(0,0,0,0,0) & 1 & (0,9,0,0,0) & 1 & (1,0,3,0,0) & 1 & (0,2,0,1,0) & 1 \\
(1,0,0,0,0) & 1 & (0,10,0,0,0) & 1 & (0,1,3,0,0) & 6 & (0,3,0,1,0) & 1 \\
(2,0,0,0,0) & 1 & (0,0,1,0,0) & 1 & (0,2,3,0,0) & 6 & (0,4,0,1,0) & 1 \\
(3,0,0,0,0) & 1 & (1,0,1,0,0) & 1 & (0,3,3,0,0) & 4 & (0,0,1,1,0) & 1 \\
(4,0,0,0,0) & 1 & (2,0,1,0,0) & 1 & (0,4,3,0,0) & 1 & (0,1,1,1,0) & 1 \\
(5,0,0,0,0) & 1 & (0,1,1,0,0) & 2 & (0,0,4,0,0) & 6 & (0,2,1,1,0) & 1 \\
(0,1,0,0,0) & 1 & (1,1,1,0,0) & 1 & (1,0,4,0,0) & 1 & (0,0,2,1,0) & 2 \\
(1,1,0,0,0) & 1 & (0,2,1,0,0) & 4 & (0,1,4,0,0) & 6 & (0,1,2,1,0) & 1 \\
(2,1,0,0,0) & 1 &(1,2,1,0,0) & 1 & (0,2,4,0,0) & 4 & (0,0,3,1,0) & 3 \\
(3,1,0,0,0) & 1 & (0,3,1,0,0) & 6 & (0,3,4,0,0) & 1 & (0,1,3,1,0) & 1 \\
(0,2,0,0,0) & 2 & (1,3,1,0,0) & 1 & (0,4,4,0,0) & 1 & (0,0,4,1,0) & 2 \\
(1,2,0,0,0) & 2 & (0,4,1,0,0) & 6 & (0,0,5,0,0) & 6 & (0,0,5,1,0) & 1 \\
(2,2,0,0,0) & 1 & (0,5,1,0,0) & 4 & (0,1,5,0,0) & 4 & (0,0,6,1,0) & 1 \\
(0,3,0,0,0) & 4 & (0,6,1,0,0) & 2 & (0,2,5,0,0) & 1 & (0,0,0,2,0) & 1 \\
(1,3,0,0,0) & 3 & (0,7,1,0,0) & 1 & (0,0,6,0,0) & 6 & (0,1,0,2,0) & 1 \\
(2,3,0,0,0) & 1 & (0,0,2,0,0) & 2 & (0,1,6,0,0) & 2 & (0,0,1,2,0) & 1 \\
(0,4,0,0,0) & 6 & (1,0,2,0,0) & 1 & (0,0,7,0,0) & 4 & (0,0,2,2,0) & 1 \\
(1,4,0,0,0) & 2 & (0,1,2,0,0) & 4 & (0,1,7,0,0) & 1 & (0,0,3,2,0) & 1 \\
(0,5,0,0,0) & 6 & (1,1,2,0,0) & 1 & (0,0,8,0,0) & 2 & (0,0,0,3,0) & 1 \\
(1,5,0,0,0) & 1 & (0,2,2,0,0) & 7 & (0,0,9,0,0) & 1 & (0,0,1,3,0) & 1 \\
(0,6,0,0,0) & 6 & (0,3,2,0,0) & 6 & (0,0,10,0,0) & 1 & (0,0,0,4,0) & 1 \\
(1,6,0,0,0) & 1 & (0,4,2,0,0) & 4 & (0,0,0,1,0) & 1 & (0,0,0,5,0) & 1 \\
(0,7,0,0,0) & 4 & (0,5,2,0,0) & 1 & (1,0,0,1,0) & 1 & (0,0,0,0,1) & 1 \\
\cline{7-8}
(0,8,0,0,0) & 2 & (0,0,3,0,0) & 4 & (0,1,0,1,0) & 1 & \textbf{TOTAL} & \textbf{209} \\
	\hline
	\end{tabular}
\end{center}
\caption{Number of inequivalent five-variable MBFs by profile.}
\label{R5}
\end{table}

As a byproduct of the calculations done for $R(7)$, we also extended the known values for the sequences $R_k(n)$ included in the OEIS. The corresponding sequences for $D_k(n)$ are \txt{A051112} to \txt{A051118} \cite{oeis}. The new values that we have computed are included in Table~\ref{table:extended}.

\begin{table}[!htbp] \footnotesize
\begin{center}
	\begin{tabular}{| c | r | r | r |} 
	\hline
	$k$ & $R_k(5)$ & $R_k(6)$ & $R_k(7)$ \\
	\hline
	2 & 13 & 22 & 34 \\
	3 & 30 & 84 & 202 \\
	4 & 49 & 287 & \textbf{1321} \\
	5 & 48 & 787 & \textbf{8626} \\
	6 & 34 & 1661 & \textbf{50961} \\
	7 & 18 & 2630 & \textbf{253104} \\
	8 & 7 & 3164 & \textbf{1025322} \\
	9 & 2 & 2890 & \textbf{3365328} \\
	10 & 2 & 2159 & \textbf{9005678} \\
	11 & 0 & 1327 & \textbf{19850932} \\
	\hline
	\end{tabular}
\end{center}
\caption[Partial list of values of $R_k(n)$]{Partial list of values $R_k(n)$, values in boldface were not known to us.}\label{table:extended}
\end{table}

\end{subsubsection}

\begin{subsubsection}{Computing Bounds on $R(n)$ and $D(n)$} \label{subsec:bounds}
Since each MBF can have at most $7! - 1 = 5039$ other functions equivalent to it, we know that $D(7)/7! \sim 479$ million is a lower bound for $R(7)$. In fact, we can increase the lower bound by looking for highly symmetric functions. For instance, the MBF with minimal term $\{1\}$ is equivalent to only six other MBFs, all with one singleton set as the only minimal term. Hence this equivalence class only has 7 functions, and increases the lower bound we have by $\frac{5040 - 7}{5040}$. If we do this for all equivalence classes of functions with at most two minimal terms, and some simple equivalence classes with three and four minimal terms, we are able to increase this lower bound by around 500. 

This raises the interesting question of what the functions in high-cardinality equivalence classes look like, that is, functions which have few or no symmetries. For $n \le 5$ functions with no symmetries are quite rare, however it appears that they already are the overwhelming majority
when $n=7$.  The number of inequivalent $n$-variable MBFs that have no symmetries, starting from $n = 1$, is the sequence 0, 1, 0, 0, 7, 7281.

\end{subsubsection}

\end{subsection}

\end{section}

\begin{section}{Implementation Details}\label{sec:compute}

All computations were done on MATLAB, a high-level scientific computing language \cite{MATLAB}. We used three computational clusters: the Optima cluster at SFU Surrey, the IRMACS computational cluster, and the bugaboo cluster of Westgrid under Compute Canada. We use MATLAB for building a prototype because it is easy to get started, and it is built for handling large vectors and matrices. It has many built-in functions that work well with the types of lists we are generating, and if desired, further work can be transported over to other programming languages.

In MATLAB, we represent MBFs as their truth table forms, $1 \times 2^n$ row vectors. This representation also lends itself well to using 32-bit integers instead of long 0-1 strings. Given an MBF of length $2^n$, we partition the zeros and ones into blocks of length 32, which we consider as a binary number $(a_0a_1a_2\ldots a_{31})_2$, and which we then convert into decimal, by computing $\sum_{k=0}^31 a_k2^k$. If $n$ is smaller than $5$, the truth table form has less than 32 entries, and we pad with zeros on the right. For $n \geq 5$, an $n$-variable MBF can be written as $2^{n-5}$ 32-bit integers.

\begin{myexample}
The six-variable MBF \small $$f = \txt{1111111011111110111111001000000011111010111010101111100000000000}$$ \rm has 64 entries, so it is divided into two blocks of 32:
\begin{align*}
\txt{11111110111111101111110010000000} & \rightarrow 20938623 \\
\txt{11111010111010101111100000000000} & \rightarrow 2053983
\end{align*}
hence the 32-bit integer representation of $f$ is $(20938623, 2053983)$. Its minimal terms are $\{1,2,4\}$, $\{3,4\}$, $\{1,5\}$, $\{2,3,5\}$, $\{1,2,3,6\}$, $\{2,4,6\}$, $\{2,5,6\}$, and $\{3,5,6\}$.
\end{myexample}

The algorithms we use involve building and frequently referencing a very long list of functions. To do this efficiently we use a hash table and a nonlinear hashing function on the 32-bit representations to perform checks and lookups quickly. In particular, we use a polynomial hash function, which acts on the four integers (say $b_1$, $b_2$, $b_3$, and $b_4$) by repeatedly adding a number $\alpha > 2$ modulo a prime $p$, and multiplying by the next component. This can be written as $b_1(\alpha + b_2(\alpha + b_3(\alpha + b_4)))$ where all operations are done modulo $p$. In our computations, using hash tables for list handling instead of binary search led to a speedup of a factor of 4 for lists of size 40000, and a factor of 8 for lists of size 1500000.

The list of functions for each profile is generated from the list for a profile which differs from it by one in a single coordinate. This allows many possibilities for traversing this lattice of profiles. Our general strategy is as follows: first, all profiles with a nonzero 1st or 6th entry, by Lemma~\ref{proflemma} can be obtained from the corresponding profile in $R(6)$. Next, we consider the remaining profiles that have a nonzero 2nd or 5th entry. It is faster to generate profile $(0,1,8,8,0,0,0)_R$ from $(0,1,8,7,0,0,0)_R$ rather than from $(0,0,8,8,0,0,0)_R$  because the 2-set in the profile $(0,1,8,7,0,0,0)$ substantially reduces the number of comparable terms we need to consider when taking  disjunctions: any 2-set is comparable to five 3-sets and ten 4-sets. Note that this also implies that the largest profiles we encounter are those solely containing 3-sets and 4-sets. In fact, the functions in these profiles account for 366689638 out of the total of 490013148 for $R(7)$, or $74.8\%$.

In the computation, we generated the list of functions for profiles with exactly one nonzero entry first, then proceeded through the list of profiles with the above considerations as a general guide. Note that the branches of computation are independent of each other and so multiple calculations can be made to run concurrently. Also many profiles were computed more than once, either as an intermediate step with the goal of computing a larger profile, or as a redundancy check. For example, the profile $(0,0,3,4,0,0,0)_R$ is computed both by a process generating profiles of the form $(0,0,3,x,0,0,0)_R$ and another one generating profiles of the form $(0,0,x,4,0,0,0)_R$. We saved lists of functions to disk for the larger profiles generated to have various points to start further computations or recover (as some jobs lasted several weeks and were susceptible to system shutdowns, etc.).

All results of computations are saved by the script files as text files, which include the numbers obtained and the computation time (Fig. \ref{fig:samplelog}). To keep track of the data, we save all the results in a database together with the computation times, with a running total (Fig. \ref{fig:sampledatabase}). At the end we obtain the number $R(7) =$ \textbf{490013148}.

\begin{figure}[htbp!]
\centering \includegraphics[width=0.9\textwidth]{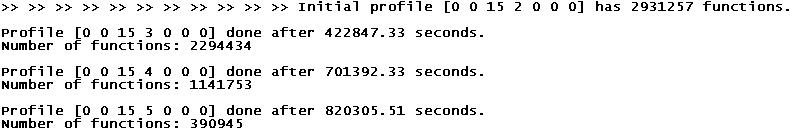}
\caption{Sample output log of the computation}\label{fig:samplelog}
\end{figure}

\begin{figure}[htbp!]
\centering \includegraphics[width=0.6\textwidth]{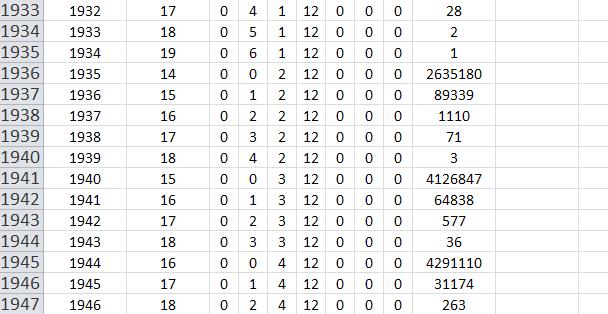}
\caption{Portion of database where results are stored}\label{fig:sampledatabase}
\end{figure}

To ensure accuracy we wrote code that outputs the minimal terms of any input function. We then checked the minimal terms of around 10,000 functions from each profile in a random sample. Moreover, each iteration of the code also includes a check to ensure that the functions in the input file indeed correspond to the profile we start with. We also tested our algorithm for $n = 6$ and obtained the correct number for both $D(6)$ and $R(6)$. As a final check, we observe that the number obtained is higher than the lower bound of 479 million discussed in Section~\ref{sec:compstrat}. Our MATLAB code and the database shown in Figure~\ref{fig:sampledatabase} are available on the website \cite{tjyusun}.

\end{section}

\begin{section}{Proof of Algorithm \ref{alg:prof}} \label{se:proof}

Here we present a proof that Algorithm \ref{alg:prof} is correct. First, we prove that the $r$th row of the matrix $K$ generated as an intermediate step in the algorithm contains lower bounds on the number of comparable $(r-1)$-sets. Then, we show that the output $P(n)$ of the algorithm contains all profiles of $n$-variable MBFs.

\begin{mylemma}\label{lemma:algK}
For $0 < r \leq n$ and $0 < x \leq {n \choose {\lfloor n/2 \rfloor}}$, the $(r,x)$-th entry of the matrix $K$ output by Algorithm~\ref{alg:prof} encodes the smallest number of $(r-1)$-sets that are comparable to any of $x$ number of $r$-sets.
\end{mylemma}
\begin{proof}
Assume first that $r = 1$, or $r - 1 = 0$. Since $K(0,x) = 0$ for any $x>0$, Line 16 of the algorithm gives
\begin{align*}
K(r,x) & = K(0,x-k-1) + {{k-1} \choose 0} \\
\Rightarrow K(r,x) & = 1.
\end{align*}
This is correct since the empty set is comparable to any number of singleton sets, and there is only one empty set.

Now assume that $r > 1$. Again, Line 16 of the algorithm gives the recursion step: $K(r,x) = K(r-1,x-{{k-1} \choose {r}}) + {{k-1} \choose {r-1}}$. We consider two cases, one where the value of $k$ was updated in Line 11 and one where it was not.

\underline{Case 1}: If $k$ was updated in Line 10, then $x$ is exactly equal to ${k \choose r}$ before $k$ was incremented by 1. This means, at Line 16, $x$ is equal to $\comb{k-1}{r}$, and hence $K(r-1,x-{{k-1}\choose r}) = K(r-1,0) = 0$. Now if we have the exact collection $\scr{U} = \comb{\{1,2,\ldots,k-1\}}{r}$, then the number of $(r-1)$-sets contained in our collection of $x$ $r$-sets must be equal to $\comb{k-1}{r-1}$, counting all $(r-1)$-subsets of $[k-1] = \{1,2,\ldots,k-1\}$.

Since any other collection of $x$ $r$-sets must contain at least $k$ elements, then the number we obtained when we considered $\scr{U}=\comb{\{1,2,\ldots,k-1\}}{r}$ must have been the lower bound for any such collection. Thus, $K(r,x) =  K(r-1,x-{{k-1} \choose {r}}) + {{k-1} \choose {r-1}}$ must be the lower bound for the number of $(r-1)$-sets contained in $x$ $r$-sets.

\underline{Case 2}: If $k$ was not updated in Line 10, then ${{k-1} \choose r} < x < {k \choose r}$. Consider the family of sets $\scr{A} = \scr{U} \cup \scr{V}$ where $\scr{U}$ is the set of all $r$-subsets of $[k-1]$, and $\scr{V}$ contains $x - \comb{k-1}{r}$ other sets, all of which contain the element $k$.

The number of $(r-1)$-sets contained in $\scr{U}$ is equal to $\comb{k-1}{r-1}$, that is, all possibilities of taking $r-1$ elements from $[k-1]$.

As for $\scr{V}$, since $\scr{U}$ already contains all $(r-1)$-subsets of $[k-1]$, we will only count the number of $(r-1)$-sets comparable to $\scr{V}$ which contain the element $k$. By removing $k$ from all of the sets in $\scr{V}$, we can see that this number is bounded below by the number of $(r-2)$-sets that must be contained in any collection of $x - \comb{k-1}{r}$ $(r-1)$-sets, or by induction, the value $K(r-1,\comb{k-1}{r})$.

Hence, the number of $(r-1)$-sets that are necessarily contained in any collection of $x$ $r$-sets must be bounded below by $K(r,x) = K(r-1,\comb{k-1}{r}) + \comb{k-1}{r-1}$, completing the proof.
\end{proof}

\begin{mytheorem}
The list $P(n)$ in Algorithm 5 contains all profiles of monotone Boolean functions on $n$ variables.
\end{mytheorem}
\begin{proof}
We perform induction on the rightmost nonzero entry in a profile vector.

When $P(n)$ is initialized, the empty profile and all profiles with a single nonzero entry in the first position are included. This is just the collection of all MBFs with only 1-sets as minimal terms, and there are $n$ such profiles as there are $n$ such sets in $2^{[n]}$.

Now assume that $P(n)$ contains all profiles with the rightmost nonzero entry in the $r$-th position.

First of all, note that when $x=0$, the conditional in Line 18 fails, and so $C(r,0) = d$, which was most recently updated to the value of $s$, the running total of all profiles so far. This means that $C(r,0) - C(r-1,0)$ is the number of new profiles added when iterating in the $(r-1)$-st row. From Line 26, we see that \txt{recent} contains exactly the profiles with rightmost nonzero entry in the $(r-1)$-st position.

The loop starting at Line 28 looks at $K(r,x)$, which by Lemma~\ref{lemma:algK} tells how many $(r-1)$-sets are equivalent to $x$ $r$-sets. Then all the profiles which have $(r-1)$-st entry \emph{at least} $K(r,x)$ will be taken -- this is \txt{candidates}. The next few lines do a substitution, replacing this number of $(r-1)$-sets by $x$ in the $r$-th entry. This new set of profile vectors is then appended into the existing list, and values are updated.

Finally we prove the fact that the variable $s$ keeps track of how many profiles have been generated already. Since $s$ is at each iteration incremented by $d$, which is in turn $C(r,x)$, we just have to prove that $C(r,x)$ encodes the number of profiles such that the rightmost nonzero entry is an $x$ in the $r$th position. But this is apparent from the loop starting at Line 18, since we are forcing $j$ to be larger than $K(r,x)$, so that from the previous row, we are only looking at profiles where the $(r-1)$-st entry is at least $K(r,x)$. This allows us to make the substitution we describe above in the loop starting at Line 28.
\end{proof}

\end{section}

\begin{section}{Conclusions and Discussion}

In this paper we propose a strategy for counting inequivalent monotone Boolean functions (MBFs), which is a challenging enumeration problem on a fundamental combinatorial object.  The strategy is to break the computation into smaller parts based on profiles of MBFs. We describe and implement a non-trivial algorithm to generate the profiles.  Using profiles, we are able to generate the full set of inequivalent $7$-variable MBFs in manageable pieces, which in particular allows us to find that the number of such functions is $R(7) = $ \textbf{490013148}.

At present it appears difficult to extend this technique to computing $R(8)$. This is because it requires generating, rather than merely counting, a nontrivial fraction of the profiles.

It is appealing to try to use $R(7)$ to compute $D(9)$, which is presently not known. Wiedemann in 1991 computed $D(8)$ using $D(6)$ and $R(6)$, by going through all pairs of functions in $D(6) \times R(6)$, and using a lookup function to calculate how many functions in $D(8)$ can be formed by fixing two ``middle functions''.   See~\cite{wiedemann} for details.

To apply this technique to the computation of $D(9)$, we would need to generate $D(7)$ from $R(7)$, store the number of functions in each equivalence class, and then calculate the number of MBFs each function contains as a preprocessing step. The difficulty lies in the sheer amount of computation needed, but as the strategy is simple to parallelize, there is some hope. To make the calculation more manageable we would like to understand the symmetries of monotone Boolean functions better. We might start by trying to count functions by their symmetry group, or by extending the sequence of inequivalent non-symmetric MBFs that we consider in Section~\ref{subsec:bounds}. 

\end{section}

\begin{section}{Acknowledgments}

This research was partially supported by an NSERC Discovery Grant. We would like to thank the SFU Math Department, the IRMACS Centre at SFU, and Westgrid and Compute Canada for the access to the computational resources needed to perform the calculations in this paper.

Also, we are grateful to Michael Monagan and Utz-Uwe Haus for discussion, in particular we thank MM for pointing out that we needed to use hash tables for the large lists.

\end{section}

\bibliographystyle{amsalpha}

\providecommand{\bysame}{\leavevmode\hbox to3em{\hrulefill}\thinspace}
\providecommand{\MR}{\relax\ifhmode\unskip\space\fi MR }
\providecommand{\MRhref}[2]{%
  \href{http://www.ams.org/mathscinet-getitem?mr=#1}{#2}
}
\providecommand{\href}[2]{#2}

\end{document}